\documentclass[conference,10pt]{IEEEtran}

\usepackage{amsmath,epsfig,amssymb,verbatim,amsopn,cite,subfigure,multirow}
\usepackage{balance}
\usepackage{footnote}
\usepackage{algorithm,algorithmic}
\usepackage[usenames,dvipsnames]{color}
\usepackage[all]{xy}  
\usepackage{url}
\usepackage{adjustbox}
\usepackage{enumitem}
\usepackage{cite,algorithm,algorithmic,amsmath,amssymb,amsthm,empheq,mhsetup}
\usepackage{subfigure,amsfonts,balance}
\usepackage{epstopdf}
\usepackage{setspace}
\usepackage[margin=16mm,top=17mm,bottom=24mm]{geometry}

  {\proof}{\proofend}
\newtheorem{proposition}{Proposition}

\newcounter{mytempeqcounter}

\newcommand{\qa}{{\bf a}}
\newcommand{\qb}{{\bf b}}

\newcommand{\qg}{{\bf g}}

\newcommand{\qq}{{\bf q}}
\newcommand{\qr}{{\bf r}}

\newcommand{\qw}{{\bf w}}
\newcommand{\qx}{{\bf x}}
\newcommand{\qy}{{\bf y}}

\newcommand{\qZ}{{\bf Z}}

\newcommand{\Sn}{\sigma_n^2}

\newcommand{\Ntx}{N}
\newcommand{\Nrx}{N}

\newcommand{\dl}{\mathtt{dl}}
\newcommand{\ul}{\mathtt{ul}}

\newcommand{\hgmkpd}{\hat{\qg}_{mk'}^{\dl}}
\newcommand{\tgmkd}{\tilde{\qg}_{mk}^{\dl}}

\newcommand{\gmkd}{\qg_{mk}^{\dl}}

\newcommand{\hgmkd}{\hat{\qg}_{mk}^{\dl}}

\newcommand{\hgmlu}{\hat{\qg}_{m\ell}^{\ul}}

\newcommand{\gmlu}{\qg_{m\ell}^{\ul}}

\newcommand{\gamdmk}{\gamma_{mk}^{\dl}}
\newcommand{\gamdmkp}{\gamma_{mk'}^{\dl}}
\newcommand{\gamuml}{\gamma_{m\ell}^{\ul}}

\newcommand{\vsl}{\varsigma_\ell}
\newcommand{\vsml}{\varsigma_{m\ell}}
\newcommand{\vsmq}{\varsigma_{mq}}

\newcommand{\betamkd}{\beta_{mk}^{\dl}}

\newcommand{\betatilmi}{\tilde{\beta}_{mi}}

\newcommand{\betakldu}{\beta_{k\ell}^{\mathtt{du}}}

\newcommand{\betamlu}{\beta_{m\ell}^{\ul}}

\title{Virtually Full-duplex Cell-Free Massive MIMO with Access Point Mode Assignment}

\author{Mohammadali Mohammadi$^\dag$, Tung T. Vu$^\dag$, Behnaz Naderi Beni$^\ddag$, Hien Quoc Ngo$^\dag$, and  Michail Matthaiou$^\dag$\\
{$^\dag$Centre for Wireless Innovation (CWI), Queen's University Belfast, BT3 9DT, Belfast, U.K.}\\
{$^\ddag$Faculty of  Engineering, Shahrekord University, Shahrekord 115, Iran}\\
{Email: \{m.mohammadi, t.vu, hien.ngo, m.matthaiou\}@qub.ac.uk,  b.naderi@stu.sku.ac.ir}
}

\allowdisplaybreaks

\begin{document}
\bstctlcite{IEEEexample:BSTcontrol}
\maketitle
\begin{abstract}
We consider a cell-free massive  multiple-input multiple-output (MIMO) network utilizing a virtually full-duplex (vFD) mode, where access points (APs) with a downlink (DL) mode and those with an uplink (UL) mode simultaneously serve DL and UL users (UEs). In order to maximize the sum spectral efficiency (SE) of the DL and UL transmissions, we formulate a mixed-integer optimization problem to jointly design the AP mode assignment and power control. This problem is subject to minimum per-UE SE requirements, per-AP power control, and per-UL UE power constraints. By employing the successive convex approximation technique, we propose an algorithm to obtain a stationary solution of the formulated problem. Numerical results show that the proposed vFD approach can provide a sum SE that is $2.5$ and $1.5$ times larger than the traditional half-duplex and heuristic baseline schemes, respectively, in terms of $95\%$-likely sum SE.

\end{abstract}

\vspace{-0em}
\section{Introduction}
Cell-free massive MIMO  has been recently put forward as an alternative for colocated massive MIMO systems, to overcome the inter-cell interference problem, while providing handover-free and balanced quality-of-services  for cell-edge UEs~\cite{Hien:cellfree}. In cell-free massive MIMO, a large number of distributed APs, that are connected to one or several central processing units (CPUs), coherently serve all UEs in the network. Cell-free massive MIMO inherits not only the favorable propagation and channel hardening from colocated massive MIMO systems, but also the macro-diversity gain from the distributed systems. Thus, it
achieves very high SE and connectivity. This has motivated the use of simple linear processing techniques in cell-free massive MIMO~\cite{Hien:cellfree}.

By relying on time-division duplex (TDD) transmission mode,  cell-free massive MIMO offers very simple control of UL and DL traffic at UEs. However, this cannot fully cater the heterogeneous UL-DL data demands from the UEs and will eventually entail high latency and network bandwidth requirements. In other words, a UE with UL (DL) data demand has to wait for a slot where APs are operating in the UL (DL) mode in order to complete its transmission.  These challenges may undermine the development of cell-free massive MIMO for wireless services with stringent latency requirements and innovative Internet-of-Things applications. Fortunately, due to the recent advances in self-interference (SI) suppression techniques~\cite{Zhang2015}, such crucial challenges have been alleviated by enabling FD capabilities at the APs to simultaneously support UL and DL data demands~\cite{tung19ICC,Nguyen:JSAC:2020,Datta:JOP:2022}. However, the system performance is still limited due to the cross-link interference, i.e., the interference received by the receiving antennas of one FD AP from the transmitting antennas of another FD AP, as well as the interference received by a DL user from another UL user. To deal with the cross-link interference, recently the concept of cell-free with network-assisted FD (NAFD) has  emerged~\cite{Wang:TCOM.2020,Xinjiang:TWC:2021}, where APs can operate in hybrid-duplex, flexible duplex, and FD modes.

For a cell-free massive MIMO system with NAFD, SI suppression at the FD APs is still an issue that makes the real application of the system challenging. Specifically, SI suppression at the FD APs entails power-hungry hardware as well as high digital signal processing cost~\cite{Zhang2015}.
In a recent work~\cite{chowdhury2021can}, a dynamic TDD scheme has been introduced in cell-free massive MIMO systems, where an in-band FD mode is virtually realized with HD APs without the SI. In particular, UL UEs are served by APs operating in an UL reception mode at the same time when DL UEs are served by another group of APs operating in a DL transmission mode. However, in \cite{chowdhury2021can}, the AP modes are scheduled by a greedy algorithm, which is not optimal. Moreover, the SE requirements for UL and DL UEs were ignored, and the impact of power control has not been investigated. Therefore, in this work, we propose a virtually FD (vFD) cell-free massive MIMO system that is similar to that in \cite{chowdhury2021can} but with  the AP mode assignment and power control jointly designed. The main contributions of our work are as follows:
\vspace{-1mm}
\begin{itemize}
	\item We formulate a problem to optimize the transmission mode assignments (i.e., UL or DL) at the HD APs and power control for maximizing the achievable sum SE in the vFD cell-free massive MIMO system. We consider the SE requirements of both UL and DL UEs, per-AP  power constraints, and per-UL UE power constraints as the optimization problem constraints.
	\item An iterative algorithm is proposed to solve the formulated mixed-integer non-convex problem. In particular, to deal with the binary mode assignment variables, we recast the problem into an approximated problem with continuous variables only. Then, we solve the problem by adopting successive convex approximation techniques.
	\item Numerical results show that our vFD approach significantly improves the SE of the system compared to the traditional HD and heuristic approaches.
\end{itemize}

\vspace{-1mm}
\textit{Notation:} We use bold upper case letters to denote matrices, and lower case letters to denote vectors. The superscripts $(\cdot)^*$, $(\cdot)^T$ and $(\cdot)^\dag$ stand for the conjugate, transpose, and conjugate-transpose, respectively.   A zero mean circular symmetric complex Gaussian distribution having variance $\Sn$ is denoted by $\mathcal{CN}(0,\Sn)$. Finally, $\mathbb{E}\{\cdot\}$ denotes the statistical expectation.

\vspace{-2mm}
\section{System model}
\vspace{-1mm}
We consider a vFD cell-free massive MIMO system, where $M$ APs serve $K$ single-antenna HD UEs in the same frequency band. Here, there are $K_u$ UL UEs and $K_d$ DL UEs, where $K=K_u+K_d$. Each AP is connected to the CPU via a high-capacity backhaul link. This system model is similar to that in~\cite{Hien:cellfree} except that we allow the DL and UL transmissions to be executed simultaneously and in the same frequency band. Each coherence block includes two phases: UL training for channel estimation and UL-and-DL payload data transmission.

Each AP is equipped with $N$ antennas, $N$ transmit radio frequency (RF) chains, and $N$ receive RF chains. The APs are able to switch between the UL or DL modes. The decision of which mode is assigned to each AP is optimized to achieve the highest sum SE of the network as discussed in Section~\ref{sec:opt}. Note that the AP mode selection is performed on the large-scale fading time scale which changes very slowly with time. For each AP $m$, we define the mode assignment variables as follows:
\vspace{-0.5em}
\begin{align}
\label{a}
a_{m} \triangleq
\begin{cases}
  1, & \text{if AP $m$ operates in the DL mode,}\\
  0, & \mbox{otherwise}.
\end{cases}, \forall m,
\\
\label{b}
b_{m} \triangleq
\begin{cases}
  1, & \text{if AP $m$ operates in the UL mode,}\\
  0, & \mbox{otherwise}.
\end{cases}, \forall m,
\end{align}
Here, we have
\vspace{-0.7em}
\begin{align}
\label{sumab}
    a_m + b_m = 1, \forall m,
\end{align}
to guarantee that the AP $m$ only operates in either the DL or UL mode.

\textbf{Uplink Channel Estimation:}
Denote by $\gmlu\in\mathbb{C}^{\Nrx \times 1}$ and $\gmkd\in\mathbb{C}^{\Ntx \times 1}$ the channel
vectors from UL UE $\ell$ and DL UE $k$ to AP $m$, respectively. Each of these channels, e.g., $\gmkd$, is modeled
in the same manner as $\gmkd=\left(\betamkd\right)^{1/2}\tgmkd$, where $\betamkd$ is the large-scale fading coefficient and $\tgmkd\in\mathbb{C}^{\Ntx \times 1}$ is the small-scale fading vector whose elements are independent and identically distributed (i.i.d.) $\mathcal{CN} (0, 1)$ random variables (RVs). In each coherence block of length $\tau_c$, each UE $\ell$ ($k$) sends the pilot sequence
of length $\tau_t$ to the APs. Assume that the pilot sequences are pairwisely orthogonal, which requires that $\tau_t\geq K_u+K_d$. At AP $m$, $\gmkd$  and $\gmlu$ are estimated by using the received pilot signals together with the minimum mean-square error (MMSE) estimation technique. By following~\cite{Hien:cellfree}, the $n$-th components of the MMSE estimates of $\gmkd$  and $\gmlu$ are as $[\hgmkd]_n$ and $[\hgmlu]_n$, respectively, which are distributed as $\mathcal{CN}(0,\gamdmk)$, and $\mathcal{CN}(0,\gamuml)$, respectively, with
\vspace{-0.3em}
\begin{align}
\gamdmk&\triangleq \mathbb{E}\left\{|[\hat{\qg}_{mk}^d]_n|^2\right\}
=
\frac{{\tau_t\rho_t}(\betamkd)^2}
{\tau_t\rho_t
\betamkd+1},
\\
\gamuml&\triangleq \mathbb{E}\left\{|[\hgmlu]_n|^2\right\}
=
\frac{{\tau_t\rho_t}(\betamlu)^2}
{\tau_t\rho_t
\betamlu
+1},
\end{align}
where $\rho_t$ is the normalized transmit power of each pilot symbol.  
We note that the AP mode selection does not affect the channel estimation.

\textbf{Downlink payload data transmission:}
Let $s_k^\dl\sim\mathcal{CN}(0,1)$ denote the intended symbol for DL UE $k$. 
Using the channel estimates from the UL training phase, the APs use maximum-ratio beamforming to transmit signals to $K_d$ DL UEs. The maximum-ratio beamforming is computationally simple for the APs to process locally~\cite{Hien:cellfree}.
The transmitted signal $\qx_{m}^{d}\in\mathbb{C}^{\Ntx\times 1}$ from AP $m$ in the DL mode is expressed by $\qx_{m}^{\dl}
= \sqrt{\rho_d}\sum_{k \in \mathcal{K}_d} \sqrt{\eta_{mk}} \left(\hgmkd\right)^*
s_{k}^{\dl},$
where $\rho_d$ is the maximum normalized transmit power at each AP and $\eta_{mk}$ is a power control coefficient. Here, we enforce
\vspace{-0.0em}
\begin{align}
\label{etaa:relation}
    (\eta_{mk} = 0, \forall k,\,\, \text{if}\,\, a_m = 0),\quad \forall m,
\end{align}
to ensure that if AP $m$ does not operate in the DL mode, all the transmit powers $\rho_d \eta_{mk}, \forall k$, at AP $m$ are zero.
Note that AP $m$ is required to meet the average normalized
power constraint, i.e., $\mathbb{E}\left\{\|\qx_{m}^{\dl}\|^2\right\}\leq \rho_d$, which can also be
expressed as the following per-AP power constraint~\cite{Hien:cellfree}:
\vspace{-0.1em}
\begin{align}
\label{DL:power:cons}
\sum_{k\in\mathcal{K}_{d}} \gamdmk \eta_{mk} \leq \frac{1}{\Ntx}.
\end{align}

\textbf{Uplink payload data transmission:}
The transmit signal  from   UL UE $\ell$ is represented by $ x_{\ell}^\ul  =\sqrt{\rho_u \tilde{\varsigma}_\ell} s_{\ell}^{\ul}$, where $s_{\ell}^{\ul}$ ($\mathbb{E}\left\{|s_{\ell}^\ul|^2\right\}=1$) and $\rho_u$ denote  respectively the transmitted symbol
and  the normalized transmit power at each UL UE. Moreover, $\tilde{\varsigma}_{\ell}$ is the transmit power control coefficient at UL UE $\ell$ with
\vspace{-0.3em}
\begin{align}
\label{UL:power:cons}
    0\leq \tilde{\varsigma}_{\ell} \leq1.
\end{align}

\textbf{Downlink SE:}
Denote by $h_{k\ell}$ the channel gain between the UL UE $\ell$ to the DL UE $k$. It can be modelled as
$h_{k\ell}=(\betakldu)^{1/2}\tilde{h}_{k\ell}$, where $\tilde{h}_{k\ell}$ is a $\mathcal{CN}(0,1)$ RV. We let
\vspace{-0.0em}
\begin{align}
\label{theta:eta:relation}
    \theta_{mk}^2 \triangleq \eta_{mk},
\end{align}
where $\theta_{mk}\geq 0$ is an additional variable. Then, the received signal at DL UE $k$ is given by
\begin{align}~\label{eq:ykdl}
y_k^{\dl}
&=
\sqrt{\rho_d}\sum_{m \in \mathcal{M}} \theta_{mk}
\left(\gmkd\right)^T\left(\hgmkd\right)^*
s_{k}^{\dl}
\nonumber\\
&\hspace{2em}+
\sqrt{\rho_d}
\sum_{m \in \mathcal{M}}
\sum_{k'\in\mathcal{K}_d \setminus k} \theta_{mk'}
\left(\gmkd\right)^T\left(\hgmkpd\right)^*
s_{k'}^{\dl}\nonumber\\
&\hspace{2em}+
\sum_{\ell\in \mathcal{K}_{u}}h_{k\ell}\sqrt{\rho_u \tilde{\varsigma}_\ell} s_{\ell}^{\ul}+w_{k}^{\dl},
\end{align}
where $w_{k}^{\dl}\sim\mathcal{CN}(0,1)$ is the AWGN at  DL UE $k$.
By invoking the use-and-then-forget capacity-bounding technique~\cite{tung19ICC, Hien:cellfree}, a closed-form expression for the achievable downlink SE can be obtained as in~\eqref{eq:DL:SE} at the top of the next page, where $\qa \triangleq \{a_m\}$, $\boldsymbol{\theta}\triangleq \{\theta_{mk}\}$, and $\tilde{\boldsymbol{\varsigma}} \triangleq \{\tilde{\varsigma}_{\ell}\}, \forall m,k,\ell$.\\
\begin{figure*}
	\begin{align}~\label{eq:DL:SE}
\mathcal{S}_{\dl,k} (\qa, \boldsymbol \theta, \tilde{\boldsymbol{\varsigma}}) =  \frac{\tau_c-\tau_t}{\tau_c}
\log_2 \left(1 +  \frac{\Ntx^2 \rho_{d}\left(
	\sum_{{m\in\mathcal{M}}}
	\theta_{mk} \gamdmk\right)^2}
{\rho_{d}\Ntx
	\sum_{k'\in\mathcal{K}_{d}}\sum_{m\in\mathcal{M}}
	\theta_{mk'}^2 \betamkd\gamdmkp
	+
	\rho_u\sum_{\ell\in\mathcal{K}_u}  \tilde{\vsl}\betakldu + 1}\right),
	\end{align}
	\hrulefill
	\vspace{-4mm}
\end{figure*}


\vspace{-3mm}
\textbf{Uplink SE:}
The interference links among the APs are modeled as Rayleigh fading channels. Let $\qZ_{mi}\in \mathbb{C}^{\Nrx\times\Ntx}$ be the channel matrix from AP $m$ to AP $i$, whose elements are  i.i.d. $\mathcal{CN}(0,\beta_{mi})$ RVs. The received signal $\qy_{m}^{\ul}\in\mathbb{C}^{\Nrx \times 1}$ at AP $m$ in the UL mode can be written as
\vspace{-0.0em}
\setcounter{equation}{11}
\begin{align}\label{eq:ymul}
\qy_{m}^{\ul}
&=
\sqrt{\rho_u}\sum_{\ell\in \mathcal{K}_{u}}\varsigma_{m\ell}\qg_{m\ell}^{\ul} s_{\ell}^{\ul}
\\
&
+
\sqrt{\rho_d}\sum_{i\in\mathcal{M}\setminus m}\sum_{k\in \mathcal{K}_d}
\sqrt{b_m \eta_{ik}}
\qZ_{mi}
(\hat{\qg}_{ik}^\dl)^*s_k^\dl
+\sqrt{b_m}\qw_{m}^{\ul},\nonumber
\end{align}
where $\varsigma_{m\ell}$ is the effective UL power control coefficients of UL UE $\ell$ at AP $m$ with
\vspace{-0.2em}
\begin{align}
    \label{varsigmaml}
    \varsigma_{m\ell}^2 &\triangleq b_m \tilde{\varsigma}_{\ell}, \quad \forall m,\ell,
\end{align}
and $\qw_{m}^{\ul}$ is the $\mathcal{CN}(0,1)$ AWGN vector. Eq. \eqref{eq:ymul} captures the fact that if AP $m$ does not operate in the UL mode, i.e., $b_m=0$, it does not receive any signal.

To detect $s_{\ell}^{\ul}$, the received signal at AP $m$  in~\eqref{eq:ymul} will be first multiplied with the Hermitian of the (locally obtained) channel estimation vector $\hgmlu$. The resulting $(\hgmlu)^\dag\qy_{m}^{\ul}$ is then forwarded to the CPU for signal detection. The aggregated received signal of UE $\ell$ at the CPU can be written as $\qr_{\ell}^{\ul}=\sum_{{m =1}}^{M}(\hgmlu)^\dag\qy_{m}^{\ul},$~\cite{tung19ICC}.
Using the use-and-then-forget capacity-bounding technique~\cite{Hien:cellfree}, we obtain the UL achievable SE of UL UE $\ell$ as in~\eqref{eq:UL:SE} at the top of the next page,  where $\qb \triangleq \{b_m\}$, $\boldsymbol \varsigma \triangleq \{\varsigma_{m\ell}\}$, and $\boldsymbol \eta \triangleq \{\eta_{mk}\}, \forall m, k, \ell$.
\begin{figure*}
\begin{align}
\label{eq:UL:SE}
	\mathcal{S}_{\ul,\ell} (\qb, \boldsymbol{\varsigma}, \boldsymbol{\eta})
	\!=\! \frac{\tau_c\!-\!\tau_t}{\tau_c}\!\log_2\!\!
	\left(\!1\!+\!\frac{\Nrx \rho_{u}\left(\sum_{\substack{m\in\mathcal{M}}}\vsml\gamuml\right)^2}
	{\!\!\rho_{u}
		\! \sum_{\substack{m\in\mathcal{M}}}
		\! \sum_{q\in\mathcal{K}_u}
		\!\!
		\vsmq^2
		\betamlu
		\gamma_{mq}^{\ul}
		\!+\!
		\rho_{d}\Ntx
		\!\sum_{\substack{m\in\mathcal{M}}}
		\!\sum_{\substack{i\in\mathcal{M}}}
		\!\sum_{k\in\mathcal{K}_d}
		\!
		b_m \eta_{ik}\betatilmi\gamuml\gamma_{ik}^{\dl}
		\!+\!
		\!\sum_{\substack{m\in\mathcal{M}}}
		\!
		b_m\gamuml}\!\right)\!\!,
\end{align}
\hrulefill
\vspace{-4mm}
\end{figure*}


\vspace{-1mm}
\section{Spectral Efficiency Maximization:\\ Problem Formulation and Solution}\label{sec:opt}
\vspace{-1mm}
In this section, we seek to design the UL and DL mode assignment vectors $\qa$ and $\qb$, and to allocate DL and UL power control coefficients $\boldsymbol \varsigma $ and $\boldsymbol \eta $, to maximize the sum SE, under the constraints on per-UE SE, transmit power at each AP, and transmit power at each UL UE. More precisely,  we formulate an optimization problem as follows:
\vspace{-0.1em}
\begin{subequations}\label{P:SE}
\begin{align}
\underset{\qx}{\max}\,\, &
\sum_{\ell\in\mathcal{K}_u} \mathcal{S}_{\ul,\ell} (\qb, \boldsymbol \varsigma, \boldsymbol \eta)   +  \sum_{k\in\mathcal{K}_d}\mathcal{S}_{\dl,k} (\qa, \boldsymbol \theta, \tilde{\boldsymbol{\varsigma}}) \\
\mathrm{s.t.} \,\,
\nonumber
& \eqref{a}-\eqref{sumab},
\eqref{etaa:relation}-
\eqref{theta:eta:relation}, \eqref{varsigmaml} \\
& \mathcal{S}_{\ul,\ell} (\qb, \boldsymbol \varsigma, \tilde{\boldsymbol{\eta}}) \geq \mathcal{S}_\ul^o,~\forall \ell\label{DL:QoS:cons}\\
&\mathcal{S}_{\dl,k} (\qa, \boldsymbol \eta, \tilde{\boldsymbol{\varsigma}}) \geq  \mathcal{S}_\dl^o,~\forall k,\label{UL:QoS:cons}
\end{align}
\end{subequations}
where $\qx\triangleq\{\qa, \qb, \boldsymbol \varsigma, \boldsymbol \theta, \boldsymbol\eta, \tilde{\boldsymbol{\varsigma}}\}$.

Let us introduce the additional variables $\{\tilde{\eta}_{imk}\}$ where
\begin{align}
    \label{tildeetaimk}
    &\tilde{\eta}_{imk} \geq b_m \eta_{ik}, \quad\forall i,m,k
    \\
    \label{tildeetaimk:2}
    &\tilde{\eta}_{imk} \leq \eta_{ik}, \quad\forall i,m,k.
\end{align}
Then, $\mathcal{S}_{\ul,\ell} (\qb, \!\boldsymbol \varsigma,\! \boldsymbol \eta)\!\! \geq\!\! \widehat{\mathcal{S}}_{\ul,\ell} (\qb, \boldsymbol \varsigma,  \tilde{\boldsymbol{\eta}}) \triangleq \!\frac{\tau_c\!-\!\tau_t}{\tau_c}\!\log_2
     \Big(\!1\!+\!\frac{\Psi_\ell^2}{\Phi_\ell}\!\Big), \forall \ell$, where $\Phi_{\ell}(\qb, \boldsymbol{\varsigma}, \tilde{\boldsymbol{\eta}}) \triangleq \rho_{u}
           \! \sum_{\substack{m\in\mathcal{M}}}
           \! \sum_{q\in\mathcal{K}_u}
                    \vsmq^2
                    \betamlu
                    \gamma_{mq}^{\ul}
\!+\!
\rho_{d}\Ntx
        \!\sum_{\substack{m\in\mathcal{M}}}
        \!\sum_{\substack{i\in\mathcal{M}}}
        \!\sum_{k\in\mathcal{K}_d}
            \tilde{\eta}_{imk} \betatilmi\gamuml\gamma_{ik}^{\dl}
\!+\!
    \!\sum_{\substack{m\in\mathcal{M}}}
                        b_m\gamuml$, $\Psi_{\ell} (\boldsymbol{\varsigma}) \triangleq \sqrt{\Nrx \rho_{u}}\sum_{\substack{m\in\mathcal{M}}}\vsml\gamuml$, and $\tilde{\boldsymbol \eta} \triangleq \{\tilde{\eta}_{imk}\}, \forall m,i,\ell,k$.
Problem \eqref{P:SE} is thus equivalent to
\vspace{-0mm}
\begin{subequations}\label{P:SE:equi}
\begin{align}
\underset{\qx,\qq_{\ul},\qq_{\dl}}{\min}\,\, &
- \sum_{\ell\in\mathcal{K}_u} q_{\ul,\ell}   -  \sum_{k\in\mathcal{K}_d} q_{\dl,k}  \\
\mathrm{s.t.} \,\,
\nonumber
& \eqref{a}-\eqref{sumab},
\eqref{etaa:relation}-\eqref{theta:eta:relation}, \eqref{varsigmaml}, \eqref{tildeetaimk}, \eqref{tildeetaimk:2} \\
& \widehat{\mathcal{S}}_{\ul,\ell} (\qb, \boldsymbol \varsigma,  \tilde{\boldsymbol{\eta}}) \geq q_{\ul,\ell}, \forall \ell  \label{UL:QoS:cons:1}\\
& q_{\ul,\ell} \geq \mathcal{S}_\ul^o, \forall \ell \label{UL:QoS:cons:2}\\
&\mathcal{S}_{\dl,k} (\qa, \boldsymbol \theta, \tilde{\boldsymbol{\varsigma}}) \geq q_{\dl,k}, \forall k \label{DL:QoS:cons:1} \\
& q_{\dl,k} \geq \mathcal{S}_\dl^o, \forall k, \label{DL:QoS:cons:2} \\
& 1 \geq \varsigma_{m\ell} \geq 0, \theta_{mk} \geq 0, \forall m,k,\ell, \label{theta:cons}
\end{align}
\end{subequations}
where $\widetilde{\qx}\triangleq \{\qx,\qq_{\ul}, \tilde{\boldsymbol \eta}, \qq_{\dl}\}$,  $\qq_{\ul} \triangleq \{q_{\ul,\ell}\}, \qq_{\dl}\triangleq\{q_{\dl,k}\}, \forall k,\ell$, are additional variables.


From the sense of \eqref{DL:power:cons}, we replace constraint \eqref{etaa:relation} by
\vspace{-0.1em}
\begin{align}
\label{etaa:relation:2}
N \gamdmk \eta_{mk} \leq a_{m}, \quad\forall m,k.
\end{align}
We also replace constraints \eqref{theta:eta:relation} and \eqref{varsigmaml} by
\begin{align}
    \label{theta:1}
    \theta_{mk}^2 &\leq \eta_{mk}, \quad\forall m,k
    \\
    \label{theta:2}
    \theta_{mk}^2 &\geq \eta_{mk}, \quad\forall m,k
    \\
    \label{varsigmaml:1}
    \varsigma_{m\ell}^2 & \leq b_m \tilde{\varsigma}_{\ell}, \quad\forall m,\ell
    \\
    \label{varsigmaml:2}
    \varsigma_{m\ell}^2 &\geq b_m \tilde{\varsigma}_{\ell}, \quad\forall m,\ell.
\end{align}
From \eqref{theta:1}
and \eqref{varsigmaml:1}, we notice that \eqref{theta:2}
and \eqref{varsigmaml:2} are respectively equivalent to
\vspace{-0.3em}
\begin{align}
    \label{C1}
    &C_1 (\boldsymbol{\theta}, \boldsymbol{\eta}) = \sum_{m\in\mathcal{M}} \sum_{k\in\mathcal{K}_d} (\eta_{mk} - \theta_{mk}^2) \leq 0.
    \\
    \label{C2}
    &C_2(\qb, \tilde{\boldsymbol{\varsigma}}, \boldsymbol{\varsigma}) = \sum_{m\in\mathcal{M}} \sum_{\ell \in \mathcal{K}_u} (b_m \tilde{\varsigma}_{\ell} - \varsigma_{m\ell}^2)  \leq 0
\end{align}

To handle the binary constraints \eqref{a} and \eqref{b}, we observe that $x\in\{0,1\}\Leftrightarrow x\in[0,1]\,\&\,x-x^2\leq0$ \cite{vu18TCOM}. Thus, \eqref{a} and \eqref{b} can be replaced by the following equivalent constraint:
\begin{align}
\label{C4}
& C_3(\qa, \qb) \triangleq \sum_{m\in\mathcal{M}} (a_{m}\!-\!a_{m}^2) + \sum_{m\in \mathcal{M}} (b_{m}\!-\!b_{m}^2) \leq 0
\\
\label{abrelax}
& 0 \leq a_{m} \leq 1,~0 \leq b_{m} \leq 1, \forall m.
\end{align}
Therefore, problem \eqref{P:SE:equi} can be written in a more tractable form as
\vspace{-0.5em}
\begin{align}\label{P:SE:equi:3}
\underset{\widetilde{\qx}\in \mathcal{F}}{\min}\,\, &
-\sum_{\ell\in\mathcal{K}_u} q_{\ul,\ell}   -  \sum_{k\in\mathcal{K}_d} q_{\dl,k}
\end{align}
where $\mathcal{F} \triangleq \{\eqref{sumab}, \eqref{DL:power:cons}, \eqref{UL:power:cons}, \eqref{tildeetaimk}, \eqref{tildeetaimk:2}, \eqref{UL:QoS:cons:1}\!-\!\eqref{theta:cons}, \eqref{etaa:relation:2}, \eqref{theta:1},\\ \eqref{varsigmaml:1},  \eqref{C1}-\eqref{abrelax}\}$ is a feasible set. Now, we consider the following problem
\vspace{-0.6em}
\begin{align}\label{P:SE:equi:relax}
\underset{\widetilde{\qx}\in \widetilde{\mathcal{F}}}{\min}\,\, &
\mathcal{L}(\widetilde{\qx})
\end{align}
where $\mathcal{L}(\widetilde{\qx})\triangleq -\sum_{\ell\in\mathcal{K}_u} q_{\ul,\ell}   -  \sum_{k\in\mathcal{K}_d} q_{\dl,k} + \lambda[\mu_1 C_1 (\boldsymbol{\theta}, \boldsymbol{\eta})
+ \mu_2 C_2(\qb, \tilde{\boldsymbol{\varsigma}}, \boldsymbol{\varsigma}) + \mu_3 C_3(\qa, \qb)]$ is the Lagrangian of \eqref{P:SE:equi:3}, $\mu_1, \mu_2, \mu_3 > 0$ are fixed weights, and $\lambda$ is the Lagrangian multiplier corresponding to constraints \eqref{C1}--\eqref{C4}. Here, $\widetilde{\mathcal{F}}\triangleq \mathcal{F}\setminus \{\eqref{C1}-\eqref{C4}\}$.

\begin{proposition}
\label{proposition-dual}
The values $C_{1,\lambda}$, $C_{2,\lambda}$, and $C_{3,\lambda}$ of $C_1$, $C_2$, and $C_3$ at the solution of \eqref{P:SE:equi:relax} corresponding to $\lambda$ converge to $0$ as $\lambda \rightarrow +\infty$. Also, problem \eqref{P:SE:equi:3} has strong duality, i.e.,
\begin{equation}\label{Strong:Dualitly:hold}
\underset{\widetilde{\qx}\in\mathcal{F}}{\min}\,\,
-\sum_{\ell\in\mathcal{K}_u} q_{\ul,\ell}   -  \sum_{k\in\mathcal{K}_d} q_{\dl,k}
=
\underset{\lambda\geq0}{\sup}\,\,
\underset{\widetilde{\qx}\in\widetilde{\mathcal{F}}}{\min}\,\,
\mathcal{L}(\widetilde{\qx}).
\end{equation}
Then, \eqref{P:SE:equi:3} is equivalent to \eqref{P:SE:equi:relax} at the optimal solution $\lambda^* \geq0$ of the sup-min problem in \eqref{Strong:Dualitly:hold}.
\end{proposition}

\begin{proof}
    The proof follows \cite{vu18TCOM}, and hence, omitted due to lack of space.
\end{proof}
\vspace{-5mm}
Note that it is theoretically required to have $C_{1,\lambda}=0$, $C_{2,\lambda}=0$, and $C_{3,\lambda}=0$ in order to obtain the optimal solution to problem \eqref{P:SE:equi:3}. According to Proposition~\ref{proposition-dual}, $C_{1,\lambda}$, $C_{2,\lambda}$, and $C_{3,\lambda}$ converge to $0$ as $\lambda\to+\infty$. For practical implementation, it is sufficient to accept $C_{j,\lambda}\leq\varepsilon, \forall j\in\{1,\dots,3\}$, for some small $\varepsilon$ with a sufficiently large value of $\lambda$.
In our numerical experiments, for $\varepsilon = 10^{-3}$, we see that $\lambda=1$ with $ \mu_1=\mu_2=0.1$, and $\mu_3=100,$ is enough to ensure that $C_{j,\lambda}\leq\varepsilon, \forall j\in\{1,\dots,3\}$. This way of selecting $\lambda$ has been widely used in the literature, e.g., see \cite{vu18TCOM} and references therein.

Problem \eqref{P:SE:equi:relax} is still difficult to solve due to non-convex constraints \eqref{UL:QoS:cons:1} and \eqref{DL:QoS:cons:1}. To deal with constraint \eqref{UL:QoS:cons:1}, we observe that the concave lower bound $\widetilde{\mathcal{S}}_{\ul,\ell} (\qb, \boldsymbol \varsigma,  \tilde{\boldsymbol{\eta}})$ of $\widehat{\mathcal{S}}_{\ul,\ell} (\qb, \boldsymbol \varsigma,  \tilde{\boldsymbol{\eta}})$ is given by \cite[Eq. (40)]{vu20TWC}
\vspace{-0.0em}
\begin{align}
    &\widetilde{\mathcal{S}}_{\ul,\ell}  (\qb, \boldsymbol \varsigma,  \tilde{\boldsymbol{\eta}}) \triangleq
    \left(\frac{\tau_c - \tau_{t}}{\tau_c\log 2}\right)
    \bigg[
          \log
             \bigg(1+\frac{(\Psi_{\ell}^{(n)})^2}{\Phi_{\ell}^{(n)}}\bigg) -\nonumber\\
    &\hspace{2em}
            \frac{(\Psi_{\ell}^{(n)})^2}{\Phi_{\ell}^{(n)}}
           		 + 2\frac{\Psi_{\ell}^{(n)}\Psi_{\ell}}{\Phi_{\ell}^{(n)}}
           		 - \frac{(\Psi_{\ell}^{(n)})^2(\Psi_{\ell}^2
           		 + \Phi_{\ell})}{\Phi_{\ell}^{(n)}((\Psi_{\ell}^{(n)})^2
           		 +\Phi_{\ell}^{(n)})}
     \Bigg].
\end{align}
Then, constraint \eqref{UL:QoS:cons:1}  is approximated by the following convex constraint
\vspace{-0.3em}
\begin{align}
    \label{UL:QoS:cons:approx}
    \widetilde{\mathcal{S}}_{\ul,\ell} (\qb, \boldsymbol \varsigma,  \tilde{\boldsymbol{\eta}}) \geq q_{\ul,\ell}, \forall \ell.
\end{align}
Now, to deal with constraint \eqref{DL:QoS:cons:1}, we see that $\mathcal{S}_{\dl,k}  (\qa, \boldsymbol \theta, \tilde{\boldsymbol{\varsigma}})$ has a concave lower bound $\widetilde{\mathcal{S}}_{\dl,k}  (\qa, \boldsymbol \theta, \tilde{\boldsymbol{\varsigma}})$ that is given as \cite[Eq. (40)]{vu20TWC}
\vspace{-0.5em}
\begin{align}
    &\widetilde{\mathcal{S}}_{\dl,k}  (\boldsymbol \theta, \tilde{\boldsymbol{\varsigma}})
     \triangleq
     \left(\frac{\tau_c - \tau_{t}}{\tau_c\log 2} \right)
     \Bigg[
      \log\bigg(1+
                  \frac{(\Xi_k^{(n)})^2}{\Omega_k^{(n)}}\bigg)
                  -\nonumber\\
                  &\hspace{3em}
                  \frac{(\Xi_k^{(n)})^2}{\Omega_k^{(n)}}
                  + 2\frac{\Xi_k^{(n)}\Xi_k}{\Omega_k^{(n)}}
                   - \frac{(\Xi_k^{(n)})^2(\Xi_k^2 + \Omega_k)}{\Omega_k^{(n)}((\Xi_k^{(n)})^2
                   	+\Omega_k^{(n)})}
       \Bigg],
\end{align}
where $\Xi_k (\boldsymbol \theta) \triangleq \Ntx \sqrt{\rho_{d}}
                           \sum_{{m\in\mathcal{M}}}
                            \theta_{mk} \gamdmk$, $\Omega_k (\boldsymbol \theta, \tilde{\boldsymbol{\varsigma}}) \triangleq \rho_{d}\Ntx
    \sum_{k'\in\mathcal{K}_{d}}\sum_{m\in\mathcal{M}}
    \theta_{mk'}^2 \betamkd\gamdmkp
+
\rho_u\sum_{\ell\in\mathcal{K}_u}  \tilde{\vsl}\betakldu + 1$.
Then, constraint \eqref{DL:QoS:cons:1}  is approximated by the following convex constraint
\vspace{-0.3em}
\begin{align}
    \label{DL:QoS:cons:approx}
    \widetilde{\mathcal{S}}_{\dl,k} (\boldsymbol \theta,  \tilde{\boldsymbol{\varsigma}}) \geq q_{\dl,k}, \forall k.
\end{align}

We also see that $xy \leq 0.25 [(x+y)^2-2(x^{(n)}-y^{(n)})(x-y) + (x^{(n)}-y^{(n)})^2]$ and $-xy \leq 0.25 [(x\!-\!y)^2\!-\!2(x^{(n)}\!+\!y^{(n)})(x\!+\!y)
+ (x^{(n)}+y^{(n)})^2], \forall x\geq0, y\geq0$ \cite{vu20TWC}. Therefore, the convex upper bounds of $C_1 (\boldsymbol{\theta}, \boldsymbol{\eta})$, $C_2(\qb, \tilde{\boldsymbol{\varsigma}}, \boldsymbol{\varsigma})$, and $C_3(\qa, \qb)$ are respectively given by
\vspace{-0.1em}
\begin{align}
    \nonumber
    &\widetilde{C}_1 (\boldsymbol{\theta}, \boldsymbol{\eta}) = \sum_{m\in\mathcal{M}} \sum_{k\in\mathcal{K}_d} \left[\eta_{mk} - 2\theta_{mk}^{(n)}\theta_{mk} + (\theta_{mk}^{(n)})^2\right]
    \\
    \nonumber
    &\widetilde{C}_2(\qb, \tilde{\boldsymbol{\varsigma}}) \triangleq \sum_{m\in\mathcal{M}} \sum_{\ell \in \mathcal{K}_u} \bigg(0.25\big[(b_m + \tilde{\varsigma}_{\ell})^2- 2(b_m^{(n)} - \tilde{\varsigma}_{\ell}^{(n)})
    \\
    \nonumber
    &\hspace{4em}
    \times(b_m \!-\! \tilde{\varsigma}_{\ell}) \!+\! (b_m^{(n)} \!- \tilde{\varsigma}_{\ell}^{(n)})^2\big]
     \!-\! 2\varsigma_{m\ell}^{(n)}\varsigma_{m\ell}\! +\! (\varsigma_{m\ell}^{(n)})^2 \bigg)
    \\
    \nonumber
    & \widetilde{C}_3(\qa, \qb) \triangleq \sum_{m\in\mathcal{M}} \left[a_{m}-2a_{m}^{(n)}a_{m} + (a_{m}^{(n)})^2\right] \nonumber\\
    &\hspace{6.0em}+ \sum_{m\in \mathcal{M}}\left[b_{m}-2b_{m}^{(n)}b_{m} + (b_{m}^{(n)})^2\right].
\end{align}
Similarly, \eqref{tildeetaimk} and \eqref{varsigmaml:1} can be approximated by the following convex constraints
\vspace{-0.3em}
\begin{align}
\label{tildeetaimk:1:convex}
    \nonumber
    & 0.25 [(b_m+\eta_{ik})^2 -2(b_m^{(n)}-\eta_{ik}^{(n)})(b_m-\eta_{ik})
    \\
    &\hspace{4em}
    + (b_m^{(n)}-\eta_{ik}^{(n)})^2] - \tilde{\eta}_{imk} \leq 0, \forall m, k 
    \\
    \label{varsigmaml:1:convex}
    \nonumber
    & \varsigma_{m\ell}^2 + 0.25 [(b_m-\tilde{\varsigma}_{\ell})^2 -2(b_m^{(n)}+\tilde{\varsigma}_{\ell}^{(n)})(b_m+\tilde{\varsigma}_{\ell})
    \\
    &\hspace{4em}
    + (b_m^{(n)}+\tilde{\varsigma}_{\ell}^{(n)})^2] \leq 0, \forall m,\ell.
\end{align}

At iteration $(n+1)$, for a given point $\widetilde{\qx}^{(n)}$, problem \eqref{P:SE:equi:relax} can finally be approximated by the following convex problem:
\begin{align}
\label{P:SE:equi:relax:approx}
\underset{\widetilde{\qx}\in\widehat{\mathcal{F}}}{\min} \,\,
& \widehat{\mathcal{L}}(\widetilde{\qx}),
\end{align}
where $\widehat{\mathcal{L}}(\widetilde{\qx}) \!=\! \!-\! \sum_{\ell\in\mathcal{K}_u} q_{\ul,\ell}   \!-\!  \sum_{k\in\mathcal{K}_d} q_{\dl,k} + \lambda[\mu_1 \tilde{C}_1 (\boldsymbol{\theta}, \boldsymbol{\eta}) 
+ \mu_2 \tilde{C}_2(\qb, \tilde{\boldsymbol{\varsigma}}, \boldsymbol{\varsigma}) + \mu_3 \tilde{C}_3(\qa, \qb)]$, $\widehat{\mathcal{F}} \triangleq \{\widetilde{\mathcal{F}}, \eqref{UL:QoS:cons:approx}, \eqref{DL:QoS:cons:approx}, \eqref{tildeetaimk:1:convex}, \eqref{varsigmaml:1:convex}\} \setminus \{\eqref{tildeetaimk}, \eqref{UL:QoS:cons:1}, \eqref{DL:QoS:cons:1}, \eqref{varsigmaml:1}\}$ is a convex feasible set. In Algorithm~\ref{alg}, we outline the main steps to solve problem \eqref{P:SE:equi:3}.
Starting from a random point $\widetilde{\qx}\in\widehat{\mathcal{F}}$, we solve \eqref{P:SE:equi:relax:approx} to obtain its optimal solution $\widetilde{\qx}^*$, and use $\widetilde{\qx}^*$ as an initial point in the next iteration. The algorithm terminates when an accuracy level of $\varepsilon$ is reached. Algorithm~\ref{alg} will converge to a stationary point, i.e., a Fritz John solution, of problem \eqref{P:SE:equi:relax} (hence \eqref{P:SE:equi:3} or \eqref{P:SE}). The proof of this fact is rather standard, and follows from \cite[Proposition 2]{vu18TCOM}.

\begin{algorithm}[!t]
\caption{Solving problem \eqref{P:SE:equi:relax}}
\begin{algorithmic}[1]
\label{alg}
\STATE \textbf{Initialize}: Set $n\!=\!0$ and choose a random point $\widetilde{\qx}^{(0)}\!\in\!\widehat{\mathcal{F}}$.
\REPEAT
\STATE Update $n=n+1$
\STATE Solve \eqref{P:SE:equi:relax:approx} to obtain its optimal solution $\widetilde{\qx}^*$
\STATE Update $\widetilde{\qx}^{(n)}=\widetilde{\qx}^*$
\UNTIL{convergence}
\end{algorithmic}
\end{algorithm}

\vspace{-2mm}
\section{Numerical Results and Discussions}~\label{Sec:Numer}
\vspace{-0mm}
We consider a cell-free massive MIMO network, where APs and UEs are randomly distributed in a square of $1 \times 1$ km${}^2$, whose edges are wrapped around to avoid the boundary effects. The distances between adjacent APs are at least $50$ m. We set $N=2$, $\mathcal{S}_\dl^o=\mathcal{S}_\ul^o=0.2$ bit/s/Hz, $K_d=K_u =5$, $\tau_c=200$,  $\tau_t=K_d+K_u$ and $\varepsilon=10^{-3}$.  We further set the bandwidth $B=20$ MHz and noise figure $F = 9$ dB. Thus, the noise power $\Sn=k_B T_0 B F=-92$ dBm, where $k_B=1.381\times 10^{-23}$ Jules/${}^o$K is the Boltzmann constant, while $T_0=290^o$K is the noise temperature. Let $\tilde{\rho}_d = 1$ W, $\tilde{\rho}_u = 0.2$~W and $\tilde{\rho}_t = 0.2$~W be the maximum transmit power of the APs, UL UEs and UL training pilot sequences, respectively. The normalized maximum transmit powers ${\rho}_d$, ${\rho}_u$, and ${\rho}_t$ are calculated by dividing these powers by the noise power. We model the large-scale fading coefficients as $\beta_{mk} = 10^{\frac{\text{PL}_{mk}^d}{10}}10^{\frac{F_{mk}}{10}}$ 
where $10^{\frac{\text{PL}_{mk}^d}{10}}$ represents the path loss, and $10^{\frac{F_{mk}}{10}}$ represents the shadowing effect with $F_{mk}\in\mathcal{N}(0,4^2)$ (in dB)~\cite{emil20TWC}.  Here, $\text{PL}_{mk}^d$ (in dB) is given by  \cite{emil20TWC}
\vspace{-0.1em}
\begin{align}\label{PL:model}
\text{PL}_{mk}^d = -30.5-36.7\log_{10}\bigg(\frac{d_{mk}}{1\,\text{m}}\bigg),
\end{align}
and the correlation among the shadowing terms from the AP $m, \forall m\in\mathcal{M}$ to different UEs $k,\ell\in\mathcal{K}_d$ ($\mathcal{K}_u$) is expressed as
\vspace{-0em}
\begin{align}\label{corr:shadowing}
\mathbb{E}\{F_{mk}F_{j\ell}\} \triangleq
\begin{cases}
  4^22^{-\delta_{k\ell}/9\,\text{m}},& \text{if $j=m$}\\
  0, & \mbox{otherwise},
\end{cases}, \forall j\in\mathcal{M},
\end{align}
where $\delta_{k\ell}$ is the physical distance between UEs $k$ and $\ell$.

\begin{figure}[t]
	\centering
	\vspace{0em}
	\includegraphics[width=0.46\textwidth]{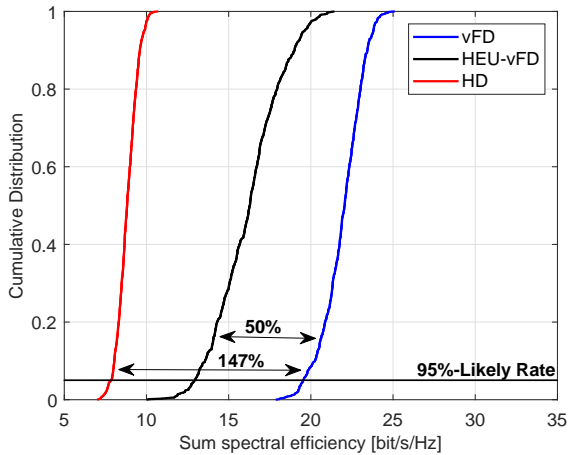}
	\vspace{-0.5em}
	\caption{ The cumulative distribution of the sum spectral efficiency ($M=30$).}
	\vspace{1.5em}
	\label{fig:Fig1}
\end{figure}
To evaluate the effectiveness of our proposed scheme (\textbf{vFD}), we consider the following baseline schemes:
\begin{itemize}
		\item \textbf{Heuristic vFD (HEU-vFD):} We apply our algorithm to optimize the per-AP power control coefficients, $\eta_{mk}$, and per-UL UE power control coefficients, $\tilde{\varsigma}_{\ell}$,  for given random mode assignment vectors $\qa$ and $\qb$, under the same SE requirement constraints for UL and DL UEs.
		\item \textbf{HD:} Data transmission phase is divided into two equal time fractions of length $(\tau_c-\tau_t)/2$ each for UL and DL data transmission. All APs serve UL and DL UEs during corresponding time fraction, with optimized $\eta_{mk}$ and $\tilde{\varsigma}_{\ell}$.
\end{itemize}

We note that \cite{chowdhury2021can} does not take into account the per-UE SE requirements. As such, the proposed greedy algorithm \cite{chowdhury2021can} needs to be modified with the per-UE SE requirements in order to have a fair comparison with our vFD approach. However, modifying this algorithm is not straightforward. Therefore, we leave a comprehensive comparison between our vFD approach and the proposed greedy algorithm in \cite{chowdhury2021can} for our future work.

Fig.~\ref{fig:Fig1} compares the cumulative distribution of the sum SE for the vFD, HEU-vFD, and HD cell-free massive MIMO systems, and for $M=30$. It is observed that the $95\%$-likely sum SE ($5\%$-outage SE) of the vFD system is about $19.47$ bits/s/Hz which is respectively  $2.5$ and $1.5$ fold higher than that of the HD system (about $7.8$ bits/s/Hz) and HEU-vFD system (about $13$ bits/s/Hz). An interesting observation here is that, the sum SE achieved by the proposed vFD  is more than $2$ times of that of the HD counterpart, i.e. more than the promised potential double gain offered by traditional FD over the HD. This can be explained by the fact that, with the proposed vFD, the APs do not suffer form the SI, and cross-link interference are efficiently managed by the AP mode assignment process.

In Fig.~\ref{fig:Fig2}, we evaluate the impact of increasing the number of AP on the performance of the proposed vFD approach. It can be observed that by increasing $M$ from $30$ to $60$, the performance of HD and HEU-vFD is slightly improved, while the gain achieved by the vFD over the
the two benchmark schemes is remarkably increased. More specifically, the gain of vFD over the HEU-vFD and HD is improved from $147\%$ to $212\%$ and $50\%$ to $80\%$, respectively.

\vspace{-1mm}
\section{Conclusion}~\label{Sec:conc}
\vspace{-0mm}
We proposed a vFD approach in cell-free massive MIMO networks, where the in-band FD is virtually realized by using the existing HD APs without the cost of hardware for SI suppression. A mixed-integer optimization problem was formulated to maximize the sum SE, with per-UE SE requirement constraints and per-UL UE and per-AP power control constraints. An iterative algorithm was proposed to solve the complicated non-convex problem. The proposed vFD approach provides significant improvement over the HD approach as well as the heuristic approach.

\begin{figure}[t]
	\centering
	\vspace{0em}
	\includegraphics[width=0.46\textwidth]{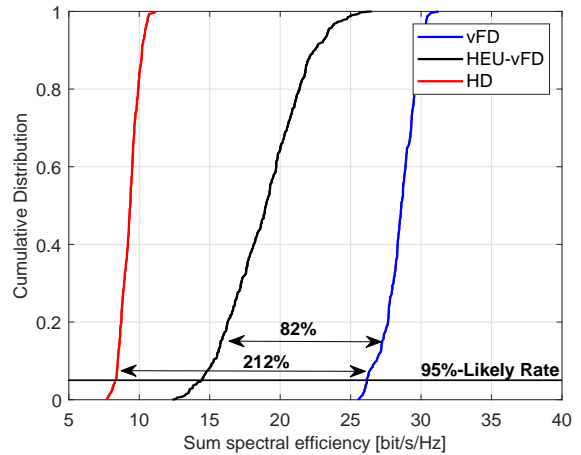}
	\vspace{-0.5em}
	\caption{ The cumulative distribution of the sum spectral efficiency ($M=60$).}
	\vspace{0.5em}
	\label{fig:Fig2}
\end{figure}

\vspace{-1mm}
\section{Acknowledgment}
This work of M. Mohammadi and M. Matthaiou was supported by a research grant from the Department for the Economy Northern Ireland under the US-Ireland R\&D Partnership Programme and by the European Research Council (ERC) under the European Union's Horizon 2020 research and innovation programme (grant agreement No. 101001331). The work of T. T. Vu and H. Q. Ngo was supported by the U.K. Research and Innovation Future Leaders Fellowships under Grant MR/S017666/1.
\vspace{-1mm}
\bibliographystyle{IEEEtran}
\bibliography{IEEEabrv,references}

\begin{thebibliography}{10}
\providecommand{\url}[1]{#1}
\csname url@samestyle\endcsname
\providecommand{\newblock}{\relax}
\providecommand{\bibinfo}[2]{#2}
\providecommand{\BIBentrySTDinterwordspacing}{\spaceskip=0pt\relax}
\providecommand{\BIBentryALTinterwordstretchfactor}{4}
\providecommand{\BIBentryALTinterwordspacing}{\spaceskip=\fontdimen2\font plus
\BIBentryALTinterwordstretchfactor\fontdimen3\font minus
  \fontdimen4\font\relax}
\providecommand{\BIBforeignlanguage}[2]{{%
\expandafter\ifx\csname l@#1\endcsname\relax
\typeout{** WARNING: IEEEtran.bst: No hyphenation pattern has been}%
\typeout{** loaded for the language `#1'. Using the pattern for}%
\typeout{** the default language instead.}%
\else
\language=\csname l@#1\endcsname
\fi
#2}}
\providecommand{\BIBdecl}{\relax}
\BIBdecl
\renewcommand{\BIBentryALTinterwordstretchfactor}{4}

\bibitem{Hien:cellfree}
H.~Q. Ngo \emph{et~al.}, ``Cell-free massive {MIMO} versus small cells,''
  \emph{IEEE Trans. Wireless Commun.}, vol.~16, no.~3, pp. 1834--1850, Mar.
  2017.

\bibitem{Zhang2015}
Z.~Zhang \emph{et~al.}, ``Full duplex techniques for {5G} networks:
  self-interference cancellation, protocol design, and relay selection,''
  \emph{IEEE Commun. Mag.}, vol.~53, no.~5, pp. 128--137, May 2015.

\bibitem{tung19ICC}
T.~T. Vu, D.~T. Ngo, H.~Q. Ngo, and T.~Le-Ngoc, ``Full-duplex cell-free massive
  {MIMO},'' in \emph{Proc. IEEE ICC}, May 2019.

\bibitem{Nguyen:JSAC:2020}
H.~V. Nguyen \emph{et~al.}, ``On the spectral and energy efficiencies of
  full-duplex cell-free massive {MIMO},'' \emph{IEEE J. Sel. Areas Commun.},
  vol.~38, no.~8, pp. 1698--1718, Aug. 2020.

\bibitem{Datta:JOP:2022}
S.~Datta \emph{et~al.}, ``Full-duplex cell-free massive {MIMO} systems:
  Analysis and decentralized optimization,'' \emph{IEEE Open J. Commun.
  Society}, vol.~3, pp. 31--50, 2022.

\bibitem{Wang:TCOM.2020}
D.~Wang \emph{et~al.}, ``Performance of network-assisted full-duplex for
  cell-free massive {MIMO},'' \emph{IEEE Trans. Commun.}, vol.~68, no.~3, pp.
  1464--1478, Mar. 2020.

\bibitem{Xinjiang:TWC:2021}
X.~Xia \emph{et~al.}, ``Joint user selection and transceiver design for
  cell-free with network-assisted full duplexing,'' \emph{IEEE Trans. Wireless
  Commun.}, vol.~20, no.~12, pp. 7856--7870, Dec. 2021.

\bibitem{chowdhury2021can}
\BIBentryALTinterwordspacing
A.~Chowdhury, R.~Chopra, and C.~R. Murthy, ``Can dynamic {TDD} enabled
  half-duplex cell-free massive {MIMO} outperform full-duplex cellular massive
  {MIMO}?'' 2021. [Online]. Available: \url{https://arxiv.org/abs/2110.09968}
\BIBentrySTDinterwordspacing

\bibitem{vu18TCOM}
T.~T. {Vu} \emph{et~al.}, ``Spectral and energy efficiency maximization for
  content-centric {C-RANs} with edge caching,'' \emph{IEEE Trans. Commun.},
  vol.~66, no.~12, pp. 6628--6642, Dec. 2018.

\bibitem{vu20TWC}
T.~T. {Vu \textit{et al}.}, ``Cell-free massive {MIMO} for wireless federated
  learning,'' \emph{IEEE Trans. Wireless Commun.}, vol.~19, no.~10, pp.
  6377--6392, Oct. 2020.

\bibitem{emil20TWC}
E.~{Bj\"{o}rnson} and L.~{Sanguinetti}, ``Making cell-free massive {MIMO}
  competitive with {MMSE} processing and centralized implementation,''
  \emph{IEEE Trans. Wireless Commun.}, vol.~19, no.~1, pp. 77--90, Jan. 2020.

\end{thebibliography}

\end{document}